\documentclass[12pt]{article}

\usepackage[T1]{fontenc}
\usepackage[a4paper]{geometry}
\usepackage{amsmath}
\usepackage{amssymb}
\usepackage{amsthm}
\usepackage[ruled, linesnumbered, vlined]{algorithm2e}
\usepackage{mathtools}
\usepackage{latexsym}
\usepackage[backend=biber,
style=alphabetic, maxnames=99, maxalphanames=99]{biblatex}
\renewbibmacro{in:}{}
\usepackage{tikz}
\usepackage{fancybox,framed}
\usepackage{authblk}

\usetikzlibrary{positioning,decorations.pathreplacing}

\DeclareFontFamily{U}{dutchcal}{\skewchar\font=45 }
\DeclareFontShape{U}{dutchcal}{m}{n}{<-> s*[1.0] dutchcal-r}{}
\DeclareFontShape{U}{dutchcal}{b}{n}{<-> s*[1.0] dutchcal-b}{}
\DeclareMathAlphabet{\mathlcal}{U}{dutchcal}{m}{n}
\SetMathAlphabet{\mathlcal}{bold}{U}{dutchcal}{b}{n}

\newtheorem{theorem}{Theorem}[section]
\newtheorem{corollary}[theorem]{Corollary}
\newtheorem{lemma}[theorem]{Lemma}
\newtheorem*{theorem*}{Theorem}
\newtheorem*{lemma*}{Lemma}
\newtheorem{definition}[theorem]{Definition}
\newtheorem{prop}[theorem]{Proposition}

\newcommand{\ket}[1]{|#1\rangle}

\def \funcSunit {\mathcal{F}^{\prime}}
\def \funcUnit {\mathcal{F}}
\DeclareMathOperator{\dist}{dist}
\DeclareMathOperator{\diag}{diag}
\DeclareMathOperator{\Lip}{Lip}

\title{Upper bounding the quantum space complexity for computing class group and principal ideal problem}
\author{Iu-Iong Ng}
\affil{Graduate School of Mathematics, Nagoya University}
\date{}

\begin{document}

\maketitle

\vspace{-3em}
\begin{abstract}
    In this paper, we calculate the upper bound on quantum space complexity of the quantum algorithms proposed by Biasse and Song (SODA'16) for solving class group computation and the principal ideal problem using the reductions to $S$-unit group computation.
    We follow the approach of Barbulescu and Poulalion (AFRICACRYPT'23) and the framework given by de Boer, Ducas, and Fehr (EUROCRYPT'20) and Eisentr\"{a}ger, Hallgren, Kitaev, and Song (STOC'14).
\end{abstract}

\section{Introduction}

\subsection{Number theoretical problems}
The computation of number theoretical problems, or computational number theory, is important both in its own right and in terms of applications, such as developing algorithms and cryptography.
The main objects in these number theoretical problems are number fields, i.e., finite extensions of the field of rational numbers $\mathbb{Q}$, and their ring of integers.
Most quantum algorithms perform exponentially better than the best-known classical algorithms for addressing several theoretical problems or their applications.
There are some problems which are believed to be hard classically but have polynomial quantum algorithms solving them, for example, integer factorisation and discrete logarithm~\cite{Shor94}, and the problems for general number fields, like the unit group computation~\cite{Hallgren_CGP, eisentrager2014quantum}, class group computation~\cite{Hallgren_CGP, biasse2016efficient}, and the principal ideal problem~\cite{Hallgren_PIP, biasse2016efficient}.
These quantum algorithms can be reduced to a type of problem, the hidden subgroup problems (HSP), which find the subgroup hidden by the period of a function.

In this work, we focus on the ideal class group computation and solving the principal ideal problem.
The ideal class group, or simply the class group, of a number field is the finite abelian group consisting of the equivalent classes of fractional ideals of the ring of integers.
One can also define the class group for an ideal of the ring of integers as the finite abelian group consisting of invertible fractional ideals of the order.
Class groups appear in many significant problems in number theory, for instance, factoring large integers and determining principal ideals in a cyclotomic field.
The computation of ideal class groups is also commonly used in other number theoretical tasks such as computing other objects of the number field like ray class group, relative class group, and unit group~\cite{biasse2014subexponential, EH10}, or in problems like finding Bach's bound on the maximum norm of the generators~\cite{bach1990explicit}.
There is also a close relation to cryptography, for example, the classical subexponential classical algorithm for integer factorisation~\cite{LP92}, and some curve-based cryptography~\cite{BCL08} that include finding relations between elements in the class group.

An ideal is principal if it can be generated by a single element.
The principal ideal problem (PIP) determines whether the input ideal is principal and finds a generator.
Like the class group computation, PIP has applications in computing ray class groups, relative class groups, unit groups, and $S$-class groups.
Problems such as lattice isomorphism and matrix similarity can be efficiently reduced to deciding whether an ideal is principal~\cite{BHJ22}.
Since many cryptographic schemes use principal ideals generated by a short element, PIP is also related to lattice-based cryptography.
Recently, it has been shown that solving PIP in polynomial time directly induces a polynomial time attack on schemes relying on the hardness of finding the short generator of a principal ideal~\cite{CDPR15}.

\subsection{Quantum algorithms for number theoretical problems}
Instead of the ordinary HSP, most recent quantum algorithms for number theoretical problems are based on a framework called the continuous hidden subgroup problem (CHSP), proposed by Eisentr\"{a}ger, Hallgren, Kitaev, and Song~\cite{eisentrager2014quantum} as a generalisation of HSP to the group $\mathbb{R}^m$ for some non-constant dimension $m$.
As an application,~\cite{eisentrager2014quantum} applied CHSP to solve the unit group problem in polynomial time by constructing an oracle that maps from a group containing the unit group to a group of ideals and then to a field of quantum states.
Recently, Barbulescu and Poulalion~\cite{barbulescu2023special} applied the complexity framework proposed by~\cite{de2020quantum} on the unit group oracle to analyse the space complexity of the unit group algorithm and propose a modified algorithm for a special case in cyclotomic fields.

\begin{theorem}[{\cite[Corollary 37]{barbulescu2023special}}]
\label{thmBP23}
    Let $K$ be a number field of discriminant $\Delta$ and unit rank $m$.
    For any error bound $\tau >0$, there exists a quantum algorithm running in time $\operatorname{poly}(m, \log |\Delta |,\log\tau)$ using a number of qubits $O(m^5+m^4\log |\Delta |)+O(m\log\tau^{-1})$ which outputs a set of generators for the unit group of $K$.
\end{theorem}

Other important applications on the CHSP are given by Biasse and Song~\cite{biasse2016efficient}, which gave a polynomial time quantum algorithms for computing class groups and solving PIP in arbitrary degree number fields by reducing the problems to a problem computing the $S$-unit group.
Notice that the previous works by Hallgren, the polynomial time algorithms for class group computation~\cite{Hallgren_CGP} and for PIP for constant degree number fields~\cite{Hallgren_PIP}, utilise the HSP algorithms but not CHSP.
The $S$-unit group problem can be viewed as a generalisation of the unit group problem with $|S|$ more parameters.
Hence,~\cite{biasse2016efficient} follows the way of defining the oracle for unit groups and extends it to the one for $S$-unit groups, which makes it possible to reduce the $S$-unit group computation to CHSP.
A well-known application of the PIP algorithm is the polynomial time quantum algorithm finding the shortest vector\footnote{The length of the shortest vector is called the minimum distance or the first successive minima.} in an ideal lattice (ideal-SVP) proposed by Cramer, Ducas, and Wesolowski~\cite{CDW21}, which applies the PIP algorithm and a variant of class group computation as dominant parts.
For the applications in cryptography, since some schemes for post-quantum cryptography rely on the hardness of PIP,~\cite{biasse2016efficient} indeed broke some cryptography systems from a theoretical point of view.
Nevertheless, a precise estimation of the quantum time or space complexity is essential in evaluating the threat raised by the quantum algorithm rather than only knowing it runs in polynomial time.

\subsection{Our results}
In this work, we follow the way in~\cite{barbulescu2023special} of calculating the space complexity of the unit group algorithm, apply the framework provided by~\cite{de2020quantum} on the oracle for $S$-unit groups, which is defined in~\cite{biasse2016efficient}, and derive the space complexity for the $S$-unit group algorithm.

\begin{theorem}
\label{thm1}
    Let $K$ be a number field of discriminant $\Delta$ and unit rank $m$, and let $S$ be a set of prime ideals such that $S=\{\mathfrak{p}_1,\dots ,\mathfrak{p}_k\}$.
    For any error bound $\tau >0$, the $S$-unit group computation algorithm from \cite{biasse2016efficient} (Theorem \ref{thmBFSU}) uses
    \[O(m^5+m^4\log |\Delta |+m^4\sum_{j=1}^k\log\mathcal{N}(\mathfrak{p}_j))+O(m\log\tau^{-1})\]
    qubits, where $\mathcal{N}(\cdot )$ is the ideal norm.
\end{theorem}

Since the unit rank $m$ has the same order as the degree $n$, we can rewrite it in the notations as~\cite{biasse2016efficient} with a maximum.

\begin{corollary}
\label{cor1}
    The $S$-unit group computation algorithm from~\cite{biasse2016efficient} (Theorem~\ref{thmBFSU}) uses $O(n^5+n^4\log |\Delta |+n^4\cdot|S|\cdot\max_{\mathfrak{p}\in S}\{\log\mathcal{N}(\mathfrak{p})\})+O(n\log\tau^{-1})$ qubits.
\end{corollary}
Our result can be viewed as a generalisation of the space complexity given by~\cite{barbulescu2023special} by taking the set $S$ to be empty.
Applying our result to the PIP quantum algorithm derived by~\cite{biasse2016efficient}, which with an input ideal $\mathfrak{a}$ runs in polynomial time in the parameters $n,\log\mathcal{N}(\mathfrak{a}),\log\Delta$, we obtain the quantum space complexity for PIP.

\begin{corollary}
\label{corPIP}
    The principal ideal problem algorithm (\cite[Theorem 1.3]{biasse2016efficient}) uses $O(n^5+n^4\log\Delta+n^4\log\mathcal{N}(\mathfrak{a}))$ qubits.
\end{corollary}

From this result, we expect the number of qubits used for PIP or its applications, e.g., ideal-SVP, to be large with respect to the degree of the input number field.
Therefore, our results indicate that, in general, implementing these algorithms may require a large-scale quantum computer.

\section{Preliminaries}

\subsection{Number field}

We follow the definitions in~\cite{biasse2016efficient}.
Let $K$ be a number field with degree $n$, i.e., $n=[K:\mathbb{Q}]$.
Denote by $n_1$ and $n_2$ the number of real embeddings and the number of pairs of complex embeddings, respectively.
Then $n=n_1+2n_2$ and the unit rank of $K$ is defined as $m=n_1+n_2-1$.
The absolute norm of an element $x\in K$ is defined as $\mathcal{N}(x)\coloneqq\prod_{\sigma}\sigma (x)\in\mathbb{Q}$, where $\sigma$ denotes the $n$ embeddings.

We denote the ring of integers of $K$ by $\mathcal{O}$.
Notice that every ring of integers of a number field is a Dedekind ring, and which implies that any ideal $I$ of $\mathcal{O}$ can be uniquely factored into a product of powers of prime ideals, i.e., $I=\prod\mathfrak{p}^{v_{\mathfrak{p}}(I)}$ with $v_{\mathfrak{p}}(I)\in\mathbb{Z}$ and only finitely many of them are non-zero.
The norm of an ideal is defined by $\mathcal{N}(I)=|\mathcal{O}/I|$, and if it is a principal ideal such that $I=(\alpha)$, then $\mathcal{N}(I)=\mathcal{N}(\alpha)$.
The unit group $\mathcal{O}^*$ consists of invertible elements, i.e., the units, in $\mathcal{O}$.
For $\alpha\in\mathcal{O}^*$, $(\alpha)=\mathcal{O}$.

\subsubsection{S-unit group}

We now define the $S$-unit group.
The definition is equivalent to the one in~\cite{biasse2016efficient}.
We refer to~\cite{neukirch2013algebraic} for more details.
For a Dedekind domain $\mathlcal{o}$, define $\mathlcal{o}(X)=\{\frac{f}{g}\, |\, f, g\in\mathlcal{o}, g\not\equiv 0\mod{\mathfrak{p}}\text{ for }\mathfrak{p}\in X\}$, where $X$ is a set of nonzero prime ideals of $\mathlcal{o}$ which contains almost all prime ideals of $\mathlcal{o}$.
Let $S=\{\mathfrak{p}_1,\dots\mathfrak{p}_k\}$ be a finite set of prime ideals of $\mathcal{O}$, and let $X_S$ be the set of all prime ideals that do not belong to $S$.
The ring $\mathcal{O}(X_S)$ has the units called the $S$-units.
If $S=\emptyset$, it turns out that $\mathcal{O}(X_S)=\mathcal{O}$, which is the special case that the $S$-units are exactly the units.\footnote{For the definition for $S$-units given by the valuations or places, it is the case that when $S=S_{\infty}$, the Archimedean valuations or infinite places, then the $S$-units are the units. See, for example,~\cite{narkiewicz2013elementary}.}
Otherwise, the $S$-units are the elements $\alpha\in K$ such that $(\alpha)=\mathfrak{p}_1^{e_1}\cdots\mathfrak{p}_k^{e_k}$ for some $e_1,\dots , e_k\in\mathbb{Z}$.
Then the $S$-units from a multiplicative group $U(S)$ and satisfy that for $\alpha\in U(S)$,
\begin{equation}
\label{eqSU}
    \alpha\cdot\mathcal{O}\cdot\mathfrak{p}_1^{-v_{\mathfrak{p}_1}(\alpha)}\cdots\mathfrak{p}_k^{-v_{\mathfrak{p}_k}(\alpha)}=\mathcal{O}.
\end{equation}

\subsubsection{$E$-ideals}

We denote $E=\mathbb{R}^{n_1}\times\mathbb{C}^{n_2}$  to be the field for $K$ under canonical embeddings, i.e., for $z\in K$, $(\sigma_1(z),\dots ,\sigma_{n_1+n_2}(z))\in E$, which is called the conjugate vector representation.
Since $\mathcal{O}$ has the structure as a $\mathbb{Z}$-lattice, its image under the embedding $K\rightarrow E$, which we denote by $\underline{\mathcal{O}}$, inherits the lattice structure and can be identified as an $\mathbb{R}^n$-lattice.
So do the fractional ideals of $\mathcal{O}$ with lattice structures correspond to fractional ideals (lattices) in $E$.

\begin{definition}[{\cite[Definition 2.1]{biasse2016efficient}}]
    An $E$-ideal is a lattice $\Lambda\subseteq E$ such that $\forall x\in\underline{\mathcal{O}}, x\Lambda\subseteq\Lambda$.
\end{definition}

Here state two theorems related to $E$-ideals that will be used in our proofs.

    \begin{theorem}[{\cite[Theorem 2.4.13]{Cohen10}}]
    \label{thmEDT}
        Let $L$ be a $\mathbb{Z}$-submodule of a free module $\Tilde{L}$ and of the same rank.
        Then there exist positive integers $d_1, \dots , d_n$ satisfying the following conditions:
        \begin{enumerate}
            \item For every $i$ such that $1\leq i<n$ we have $d_{i+1}|d_i$.
            \item             $[\Tilde{L}: L]=d_1\cdots d_n$.
            \item There exists a $\mathbb{Z}$-basis $(v_1, \dots v_n)$ of $\Tilde{L}$ such that $(d_1v_1, \dots ,d_nv_n)$ is a $\mathbb{Z}$-basis of $L$.
        \end{enumerate}
        Furthermore, the $d_i$ are uniquely determined by $L$ and $\Tilde{L}$.
    \end{theorem}

\begin{theorem}[{\cite[Theorem 4.7.4]{Cohen10}}]
    \label{thmHNF}
    Let $M$ be a module with denominator 1 with respec to a given $R$ (i.e. $M\subset R$), and $W=(w_{i, j})$ its Hermite normal form (HNF) with respect to a basis $\alpha_1,\dots ,\alpha_n$ of $R$.
    Then the product of the $w_{i, i}$ (i.e. the determinant of $W$) is equal to the index $[R:M]$.
\end{theorem}

\subsection{The unit group computation algorithm}

In this subsection, we review the ideas for the unit group algorithm from~\cite{eisentrager2014quantum}.
By the properties of units, one can identify $\mathcal{O}^*$ as a subgroup of $\hat{G}=\mathbb{R}^{n_1+n_2-1}\times\mathbb{Z}_2^{n_1}\times(\mathbb{R}/\mathbb{Z})^{n_2}$.
To see this, we consider the mapping $\varphi :\hat{G}\rightarrow E$ translating between the log coordinates and the conjugate vector representation.
\[\begin{aligned}
\varphi : &(u_1,\dots , u_{n_1+n_2},\mu_1,\dots ,\mu_{n_1},\theta_1,\dots ,\theta_{n_2})\\
&\mapsto ((-1)^{\mu_1}e^{u_1},\dots , (-1)^{\mu_{n_1}}e^{u_{n_1}}, e^{2\pi i\theta_1}e^{u_{n_1+1}},\dots , e^{2\pi i\theta_{n_2}}e^{u_{n_1+n_2}}).
\end{aligned}\]
Since the units are the elements $z\in\mathcal{O}$ with $\mathcal{N}(z)=\pm 1$, for a unit written as $z=e^{\boldsymbol{u}}\boldsymbol{v}$, where $\boldsymbol{u}\in\mathbb{R}^{n_1+n_2}$, it satisfies that $\sum_{j=1}^{n_1+n_2}u_j=0$, and hence $\mathbb{R}^{n_1+n_2-1}$ is enough for the presentation of units.

The oracle function defined in~\cite{eisentrager2014quantum} is a composition of two mappings:
\[\funcUnit:\hat{G}\xrightarrow{f_c}\{ E\text{-ideals}\}\xrightarrow{f_q}\{\text{quantum states}\},\]
where $f_q$ encodes a lattice $L$ into a quantum state $\ket{L}$ so that it provides a canonical representation for lattices, and $f_c$ map an element of $\hat{G}$ to the principal ideal generated by it.

\subsection{The $S$-unit group computation algorithm}

\begin{theorem}[{\cite[Theorem 1.1]{biasse2016efficient}}]
\label{thmBFSU}
    There is a quantum algorithm for computing the $S$-unit group of a number field $K$ in compact representation which runs in polynomial time in the parameters $n=\deg(K),\log(|\Delta|),|S|$ and $\max_{\mathfrak{p}\in S}\{\log(\mathcal{N}(\mathfrak{p}))\}$, where $\Delta$ is the discriminant of the ring of integers of $K$.
\end{theorem}

One of the contributions by~\cite{biasse2016efficient} is showing how to get an exact compact representation of the desired field element, which is processed classically.
Notice that similar to the unit group, the $S$-unit group can be identified as a subgroup of $G=\hat{G}\times\mathbb{Z}^{|S|}=\mathbb{R}^{n_1+n_2-1}\times\mathbb{Z}_2^{n_1}\times(\mathbb{R}/\mathbb{Z})^{n_2}\times\mathbb{Z}^{|S|}$, where $\mathbb{Z}^{|S|}$ corresponds the valuations (exponents) of the prime ideals in $S$.
The algorithm for Theorem \ref{thmBFSU} applies the CHSP framework from \cite{eisentrager2014quantum} with an HSP oracle
\[\funcSunit :G\xrightarrow{f_c^{\prime}} \{ E\text{-ideals}\}\xrightarrow{f_q}\{\text{quantum states}\},\]
where $f_c^{\prime}$ is defined as
\[f_c^{\prime}(\boldsymbol{y}, v_1,\dots , v_{|S|})=\varphi(\boldsymbol{y})\cdot\underline{\mathcal{O}}\cdot\mathfrak{p}_1^{-v_1}\cdots\mathfrak{p}_{|S|}^{-v_{|S|}},\]
and $f_q$ is defined as the one in~\cite{eisentrager2014quantum}, which can be extended from $E$-integral ideals to $E$-fractional ideals.
By the property of $S$-units, $\funcSunit$ hides the subgroup of $G$, denoted by $U(S)$, identified as the $S$-unit group.
The periodicity of $f_c^{\prime}$ on $U(S)$ is proved by Proposition 5.1 in~\cite{biasse2016efficient}.
We rephrase it as follows.

\begin{prop}[{\cite[Proposition 5.1]{biasse2016efficient}}]
    \label{propperiodic}
    For any $(y, (v_j))$ and $(y^\prime, (v_j^\prime))$, let $(u, (w_j))=(y^\prime, (v_j^\prime))-(y, (v_j))$.
    Then the function $f_c^{\prime}$ satisfies that
    \[f_c^{\prime}(y^\prime, (v_j^\prime))=f_c^{\prime}(y, (v_j))\Leftrightarrow\varphi(u)\in U(S).\]
    In particular, $v_{\mathfrak{p}_j}(\varphi(u))=w_j, \forall j=1,\dots , |S|$ if $\varphi(u)\in U(S)$.
\end{prop}

\paragraph{Distances.}

In Section A.3 of~\cite{eisentrager2014quantum}, the distance between two lattices is defined as the geodesic distance on the group $\operatorname{GL}_n(\mathbb{R})$ between their bases matrices $B$ and $B^\prime$.
Here, to specify the distance by lattice but not its bases, we modify the definition as stated below.

\begin{definition}
    \label{defDistg}
    \[\dist_g(L, L^\prime)\coloneqq\inf\{\Vert A\Vert_2: e^AB_{L^\prime}=B_L, B_L\text{ and }B_{L^\prime}\text{ are bases for }L\text{ and }L^\prime\text{, respectively}\},\]
    where $\Vert (a_{jk})\Vert_2=\sqrt{\sum_{j, k}|a_{jk}|^2}$.
\end{definition}

On the other hand, the definition of the distance for the elements in the domain in~\cite{biasse2016efficient} and for the lattices are defined as follows.

\begin{definition}[{\cite[Definition 5.2]{biasse2016efficient}}]
    Let $(z, (v_j)_{j\leq |S|})$ and $(z^\prime, (v^\prime_j)_{j\leq |S|})$, we define their distance in $G/U(S)$, $\dist_{G/U(S)}((z, (v_j)), (z^\prime, (v^\prime_j)))$, by
$$\inf\{\Vert a\Vert +\sum_j|w_j|e_j\log(p_j)\text{ such that }(z^\prime, (v^\prime_j)_{j\leq |S|})=(z, (v_j)_{j\leq |S|})+(a, (w_j))+u, u\in U(S)\},$$
where $\Vert a\Vert$ is the Euclidean norm of the vector corresponding to $a$ in $\mathbb{R}^{n_1+n_2}\times\mathbb{Z}_2^{n_1}\times(\mathbb{R}/\mathbb{Z})^{n_2}$.
The $p_j, e_j$ are defined as $\mathcal{N}(\mathfrak{p}_j)=p_j^{e_j}$.
\end{definition}

\begin{definition}[{\cite[Definition 5.3]{biasse2016efficient}}]
    $$\dist(L, L^\prime)=\inf\left\{\Vert a\Vert+\sum_j\log(d_j)+n\log(d)\text{ such that }L_\Delta=e^{\diag(a_j)}B_\omega\diag(d_j/d)\right\},$$
where $L_\Delta$ runs over all the matrices of a basis of $L^\prime/L$ such that there is a matrix $B_\omega$ of an integral basis of $\mathcal{O}$, $d_j, d\in\mathbb{Z}_{>0}$, and $\Vert a\Vert$ is the Euclidean norm of the vector $a\in \mathbb{R}^{n_1+n_2}\times\mathbb{Z}_2^{n_1}\times(\mathbb{R}/\mathbb{Z})^{n_2}$ corresponding to $(a_j)_{j\leq n}\in E$ satisfying $L_\Delta=e^{\diag(a_j)}B_\omega\diag(d_j/d)$.
\end{definition}

The definition for $L^\prime/L$ is not explicitly stated in~\cite{biasse2016efficient}, but it can be realised from the statements and proofs in the paper as $L^\prime/L=\varphi(z^\prime-z)\underline{\mathcal{O}}\prod\mathfrak{p}_j^{-(v_j^\prime-v_j)}$.

\subsection{Complexity of the CHSP framework}

The complexity of the CHSP framework from~\cite{eisentrager2014quantum} is studied in~\cite{de2020quantum, barbulescu2023special}.
An HSP oracle hiding $L$ on $\mathbb{R}^m$ for some positive integer $m$ is defined as follows.

\begin{definition}[{\cite[Definition 1.1]{eisentrager2014quantum}}]
\label{defHSP}
    A function $f:\mathbb{R}^m\rightarrow H$, where $H$ is the set of unit vectors in some Hilbert space, is said to be an $(a, r,\epsilon)$-HSP oracle of the full-rank lattice $L\subset\mathbb{R}^m$ if
    \begin{enumerate}
        \item $f$ is periodic on $L$;
        \item $f$ is $a$-Lipschitz;
        \item For all $x, y\in\mathbb{R}$ such that $\dist_{\mathbb{R}^m/L}(x, y)\geq r$, it holds that $|\langle f(x)\ket{f(y)}|\leq\epsilon$.
    \end{enumerate}
\end{definition}

Theorem~\ref{thmBP23} from~\cite{barbulescu2023special} calculates the space complexity for the unit group algorithm from~\cite{eisentrager2014quantum}.
According to~\cite[Lemma 21, Lemma 34 and Corollary 37]{barbulescu2023special}, one can obtain that $O(\log(1/\lambda_1^*))=O(m+\frac{1}{m}\log D)$, where $\lambda_1^*$ is the first successive minima of the dual lattice of $L$, and the following corollary on the complexity of the CHSP framework applied on an HSP oracle $f$.

\begin{corollary}[\cite{barbulescu2023special}]
\label{corspcpxty}
    Given access to an HSP oracle $f$, the CHSP algorithm uses
    \[O(m^3\log\Lip(f)+m^4\log\Delta)+O(m\log\tau^{-1})\]
    qubits.
\end{corollary}

In~\cite{de2020quantum}, the algorithm of~\cite{eisentrager2014quantum} is rigorously analysed (and modified) as follows. The oracle function is considered on a restricted domain $\mathbb{D}^m\coloneqq\left(\frac{1}{q}\mathbb{Z}^m\right)/\mathbb{Z}^m$ with the parameter $q$ relating to the space complexity, and $H$ is a subset of a Hilbert space of dimension $2^n$.
The input state is a Gaussian superposition over the representatives $x\in\mathbb{D}^m_{\text{rep}}\coloneqq\frac{1}{q}\mathbb{Z}^m\cap [-\frac{1}{2},\frac{1}{2})^m$ of $\mathbb{D}^m$ with the parameter $s\in\mathbb{R}$.
The algorithm has oracle access to $f$ which maps $\ket{x}\ket{0}$ to $\ket{x}\ket{f(Vx)}$ with the parameter $V\in\mathbb{R}$.
The first register uses $m\log q$ qubits, and the second register uses $N$ qubits.
Reference~\cite{de2020quantum} denotes that $\log q\eqqcolon Q$ and gives the number of qubits needed for the oracle as follows.

\begin{theorem}[{\cite[Theorem 2]{de2020quantum}}]
   There exists dual lattice sampler quantum algorithm with the error parameter $\eta >0$ and the relative distance parameter $1/2>\delta >0$ which uses one quantum oracle call to $f$, $Qm+N$ qubits, where
   \begin{equation}
   \label{eqcpxty}
       Q=O\left(m\log\left(m\log\frac{1}{\eta}\right)\right)+O\left(\log\left(\frac{\Lip(f)}{\eta\delta\lambda_1^*}\right)\right).
   \end{equation}
\end{theorem}

Later in~\cite[Corollary 37]{barbulescu2023special}, the number of qubits for the second register is claimed that one stores the values of $f$ on $Qm$ qubits.

\paragraph{Lipschitz constant.}

Consider the quantum encoding part, $f_q$, in the oracle functions for both the unit group algorithm and the $S$-unit group algorithm.
The Lipschitz continuity for $f_q$ is proven as stated below.

\begin{theorem}[{\cite[Theorem D.4]{eisentrager2014quantum}}]
\label{thmLip}
    $\Vert\ket{f_q(L)}-\ket{f_q(L^\prime)}\Vert\leq\Lip(f_q)\cdot\dist_g(L, L^\prime)$.
\end{theorem}

From this result,~\cite{eisentrager2014quantum} derived the value of $\Lip(\funcUnit )$, and the order of $\log\Lip(\funcUnit )$ is given in~\cite{barbulescu2023special}.

\begin{theorem}[{\cite[Theorem 36]{barbulescu2023special}}]
\label{thmlogLip}
    $\log_2\Lip(\funcUnit )=O(m^2+m\log\Delta)$.
\end{theorem}

\section{Proof for the main theorem}

 To prove Theorem~\ref{thm1}, we need the Lipschitz constant for $\funcSunit$.
 Since the Lipschitz constants depend on the distance chosen, the approach is to calculate the Lipschitz constant for $f_c^{\prime}$ with the distance $\dist(\cdot,\cdot)$ between the ideals (Lemma \ref{lem2}), and the relation between distances $\dist(\cdot,\cdot )$ and $\dist_g(\cdot,\cdot )$ (Lemma \ref{lem1}).
 Combined with Theorem \ref{thmLip}, one can obtain an upper bound for $\log\Lip(\funcSunit )$ from the upper bound for $\log\Lip(\funcUnit )$.
 Therefore, we will use the following two lemmas shown at the end of this section.

\begin{lemma}
    \label{lem1}
    For any $E$-ideals $L$ and $L^\prime$, $\dist_g(L, L^\prime)=O(n^{2n+2}+\prod_j\mathcal{N}(\mathfrak{p}_j)^{c_jn})\cdot\dist(L, L^\prime)$ holds for some constants $c_j$.
\end{lemma}

     \begin{lemma}
\label{lem2}
    For any $x, y\in G$ and $L=f_c^\prime(x), L^\prime=f_c^\prime(y)$,
    \[\dist(L, L^{\prime})=O(n)\cdot\dist_{G/U(S)}(x, y)\]
    holds.
\end{lemma}

\begin{proof}[Proof for Theorem~\ref{thm1}]
    Combining Theorem~\ref{thmLip}, Lemma~\ref{lem1} and Lemma~\ref{lem2}, we have that
    \[\Vert\ket{f_q(x)}-\ket{f_q(y)}\Vert= O\left(n^{2n+3}+\prod_j\mathcal{N}(\mathfrak{p}_j)^{c_jn}\right)\cdot\Lip (f_q)\cdot\dist_{G/U(S)}(x, y)\]
    for some constants $c_j$, which implies that
    \[\Lip(\funcSunit )=O\left(m^{2m+3}+\prod_j\mathcal{N}(\mathfrak{p}_j)^{c_jm}\cdot\Lip(\funcUnit )\right),\]
    and hence
    \[\log\Lip(\funcSunit )=O\left(m^2+m\log\Delta+m\sum_{j=1}^{|S|}\log\mathcal{N}(\mathfrak{p}_j)\right)\]
    by Theorem~\ref{thmlogLip}.
    
    Notice that in~\cite[Theorem 5.4]{biasse2016efficient}, it is shown that the $(\Lip(\funcSunit), r,\epsilon)$-HSP oracle $\funcSunit$ reduces to an $(\Lip(\hat{\funcSunit}), \hat{r},\epsilon)$-HSP oracle $\hat{\funcSunit}$ defined on $\mathbb{R}^{\hat{n}}$ with $\hat{n}=2(n_1+n_2)+|S|-1$, and moreover, $O(\log\Lip(\hat{\funcSunit}))=O(\log\Lip(\funcSunit))$.
    In order to apply the CHSP framework (\cite[Theorem 1]{de2020quantum}) on $\hat{\funcSunit}$, one needs that $\epsilon <1/4$.
    According to~\cite[Theorem 36]{barbulescu2023special}, we can take tensor product $\otimes^c\hat{\funcSunit}$ with a constant $c$ large enough such that the resulting oracle hides the same lattice.
    It is a $(c\Lip(\hat{\funcSunit}), \hat{r}, \epsilon^c)$-HSP oracle satisfying that $\epsilon^c<1/4$.
    Hence by applying Corollary~\ref{corspcpxty} with $f=\mathcal{F}^{\prime}$, we obtain the number of qubits 
    \[O(m^5+m^4\log\Delta+m^4\sum_{j=1}^{|S|}\log\mathcal{N}(\mathfrak{p}_j))+O(m\log\tau^{-1})\]
    as claimed.
\end{proof}

Below, we give the proofs of the two lemmas.

\begin{proof}[Proof for Lemma~\ref{lem1}]
Fix $L=\varphi (z)\underline{\mathcal{O}}\prod\mathfrak{p}_j^{-v_j}$ and $L^\prime =\varphi (z^\prime)\underline{\mathcal{O}}\prod\mathfrak{p}_j^{-v_j^\prime}$.
Let $a, d_j, d$ be ones that satisfy $\dist(L, L^\prime)=\Vert a\Vert+\sum_j\log(d_j)+n\log(d)$.
We first consider the special case by assuming that $d_j=d=1$ for all $j$ such that $L_\Delta=e^{\diag(a_j)}B_\omega$.
Therefore, without loss of generality, we can write $L^\prime/L=\varphi(z^\prime-z)\underline{\mathcal{O}}\prod\mathfrak{p}_j^{-(v_j^\prime-v_j)}$, where $-(v_j^\prime-v_j)\geq 0$ for all $j$.
It is implied that $1/\varphi(z)L\supseteq L^\prime$, and that there exist the HNF-basis $H$ for $L^\prime$ and a matrix $A$ satisfying $e^A=He^{\diag((-z)_j)}$ such that
\[\begin{aligned}
    \dist_g(L, L^\prime) & \leq\Vert A\Vert_2\\
    & \leq n^2\cdot\det(H)\cdot\dist(L, L^\prime)\\
    & \leq n^2\prod\mathcal{N}(\mathfrak{p}_j)^{c_j}\cdot\dist(L, L^\prime)
\end{aligned}\]
by Theorem \ref{thmHNF} for some constants $c_j$.

Now suppose that not all of $d_j$ and $d$ are ones, so that $\sum_j\log(d_j)+n\log(d)\geq\log 2$.
Since $L$ and $L^\prime$ are fractional ideals of $\mathcal{O}$, there exist $\alpha ,\alpha^\prime\in K$ such that they can be written as $L=\frac{1}{\alpha}M$ and $L^\prime=\frac{1}{\alpha^\prime}M^\prime$ for some integral ideals $M, M^\prime$ in $\mathcal{O}$.
Let $W, W^\prime\in\operatorname{GL}_n(\mathbb{Z})$ denote the HNF-bases such that $\frac{1}{\alpha}WB_\omega$ and $\frac{1}{\alpha^\prime}W^\prime B_\omega$ are basis for $L$ and $L^\prime$, respectively.
Under the condition, we can obtain an upper bound for the matrix $A$ satisfying that $e^A=\alpha/\alpha^\prime W^\prime W^{-1}$ by Theorem \ref{thmHNF} and the inequality $\dist(L, L^\prime)\geq\log 2$ derived from the above.
\[\begin{aligned}
    \dist_g(L, L^\prime) & \leq\Vert A\Vert_2\\
    & \leq\left\Vert\frac{\alpha}{\alpha^\prime}\right\Vert\Vert W^\prime\Vert_2\Vert W^{-1}\Vert_2\\
    & \leq n^2\prod_j\mathcal{N}(\mathfrak{p}_j)^{h_j}\sqrt{\sum_jj|w_{j, j}^\prime|^2}\cdot\frac{\Vert W\Vert_2^{n-1}}{|\det W|}\\
    & \leq n^2\prod_j\mathcal{N}(\mathfrak{p}_j)^{h_j}\frac{n(n+1)}{2}\prod_jw_{j, j}^\prime\left(\frac{n(n+1)}{2}\prod_jw_{j, j}\right)^{n-1}\\
    & \leq\frac{n^2}{\log 2}\left(\frac{n(n+1)}{2}\right)^n\prod_j\mathcal{N}(\mathfrak{p}_j)^{\Tilde{h}_jn}\cdot\dist(L, L^\prime)
\end{aligned}\]
for some constants $h_j$ and $\Tilde{h}_j$.

Hence we can obtain that $\dist_g(L, L^\prime)=O(n^{2n+2}+\prod_j\mathcal{N}(\mathfrak{p}_j)^{\Tilde{h}_jn})\cdot\dist(L, L^\prime)$ as claimed.
\end{proof}

\begin{proof}[Proof for Lemma~\ref{lem2}]
    Fix $L=\varphi (z)\underline{\mathcal{O}}\prod\mathfrak{p}_j^{v_j}$ and $L^\prime =\varphi (z^\prime)\underline{\mathcal{O}}\prod\mathfrak{p}_j^{v_j^\prime}$ that are the images of $(z, (v_j))$ and $(z^\prime , (v_j^\prime))$ under the map $f_c^\prime$, respectively.
    We write
    \[\dist_{G/U(S)}((z, (v_j)), (z^\prime , (v_j^\prime)))=\Vert z-z^\prime-u\Vert +\sum |v_j-v_j^\prime-w_j|e_j\log p_j\]
    for some $(u, (w_j)))\in U(S)$.
    Notice that if $(z-z^\prime-u, (v_j-v_j^\prime-w_j))=0$, i.e., $(z, (v_j))$ and $(z^\prime , (v_j^\prime))$ satisfy the condition in Proposition~\ref{propperiodic} that they are different by an $S$-unit, then $L=L^\prime$, and moreover, $L^\prime/L=L/L^\prime=\underline{\mathcal{O}}$, which implies that
    \[\dist(L, L^\prime)=\dist_{G/U(S)}((z, (v_j)), (z^\prime , (v_j^\prime)))=0.\]

    Now we suppose that $(z-z^\prime-u, (v_j-v_j^\prime-w_j))\neq 0$ and that $z, z^\prime, v_j, v_j^\prime$ satisfy the infimum, i.e., $\dist_{G/U(S)}((z, (v_j)), (z^\prime , (v_j^\prime)))=\Vert z-z^\prime\Vert +\sum |v_j-v_j^\prime|e_j\log p_j$.
    From the definition of the distance $\dist(\cdot,\cdot)$ among $E$-ideals, we have that
    \[\dist(L, L^\prime)\leq\Vert z-z^\prime\Vert +\sum_j\log d_j+n\log d\]
    for some $d_j, d\in\mathbb{Z}$ such that $\prod d_j=\prod\mathcal{N}(\mathfrak{p}_j)^{\min\{-(v_j^\prime-v_j), 0\}}$ and $d\leq\prod\mathcal{N}(\mathfrak{p}_j)^{\max\{-(v_j^\prime-v_j), 0\}}$ by Theorem~\ref{thmEDT}.
    Therefore, the upper bound for the distance can be taken as
    \[\dist(L, L^\prime)\leq\Vert z-z^\prime\Vert +n\sum_jc_j\log\mathcal{N}(\mathfrak{p}_j)\]
    for some constants $c_j$.
    Then we can derive that
    \[\begin{aligned}
    \dist(L, L^\prime) & \leq\Vert z-z^\prime\Vert +n\sum_jc_je_j\log p_j\\
    & \leq nh_j\cdot\dist_{G/U(S)}((z, (v_j)), (z^\prime , (v_j^\prime))),
    \end{aligned}\]
    where $h_j$ are constants.
    Hence, it follows that
    \[\dist(L, L^\prime)=O(n)\cdot\dist_{G/U(S)}((z, (v_j)), (z^\prime , (v_j^\prime)))\]
    as claimed.
\end{proof}

\section{Reductions to $S$-unit computation}

Reference~\cite{biasse2016efficient} proposed two problems that can be reduced to $S$-unit group computation, the class group problem, and the principal ideal problem.

\subsection{Class group problem}

In the class group algorithm from~\cite{biasse2016efficient}, the set is taken as
\[S_{\text{CGP}}=\{\mathfrak{p}:\text{ prime }\mid\mathcal{N}(\mathfrak{p})\leq 48(\log |\Delta|)^2\},\]
and the output is the Smith normal form (SNF)  of the valuations of the result of $S_{\text{CGP}}$-unit group computation.
The computation of the SNF is classical, so for the quantum space complexity, it suffices to evaluate $|S_{\text{CGP}}|$.
We approximate the number of rational primes smaller than $48(\log |\Delta|)^2$ with the following theorem, which was first discovered by Gauss; see, for example, \cite{Apostol_ANT}.
Denote by $\pi (x)$ the number of rational primes less than $x$.

\begin{theorem}[The prime number theorem]
    $\pi (x)\sim\frac{x}{\log x}$ as $x\rightarrow\infty$.
\end{theorem}

Let $e=\left\lfloor\frac{\log x}{\log 2}\right\rfloor$ and $p$ denotes a rational prime number.
Then the number of rational prime powers that are smaller than $x$ is
\[\begin{aligned}
    \sum_{p\leq 1}1+\cdots +\sum_{p^e\leq x}1=\frac{x}{\log x}+\frac{x^{\frac{1}{2}}}{\frac{1}{2}\log x}+\cdots +\frac{x^{\frac{1}{e}}}{\frac{1}{e}\log x}
\end{aligned}\]
and has the order $O(x/\log x)$.
From Corollary~\ref{cor1}, we can obtain the quantum space complexity of the class group algorithm.

\begin{corollary}
    Under the Generalised Riemann Hypothesis, the class group computation algorithm (\cite[Theorem 1.2]{biasse2016efficient}) uses 
    $O(n^5+n^4(\log\Delta)^4/\log\log\Delta)+O(n\log\tau^{-1})$
    qubits.
\end{corollary}

\subsection{Principal ideal problem}

In the algorithm reducing PIP to $S$-units from~\cite{biasse2016efficient}, the first step, that factors the input ideal to the product of powers of prime ideals $\mathfrak{a}=\mathfrak{p}_1^{a_1}\cdots\mathfrak{p}_k^{a_k}$, is quantum.
The set $S$ is taken to be the prime ideals dividing the ideal, i.e., $S_{\text{PIP}}=\{\mathfrak{p}_1,\dots ,\mathfrak{p}_k\}$.
The last step following the $S_{\text{PIP}}$-unit group computation is to classically solve equations on the valuations.

For the ideal factorisation, we follow the algorithm from~\cite[Algorithm 2]{EH10}, which shows that factoring integers is the only computationally difficult part.

\begin{prop}[{\cite[Lemma 4.1]{EH10}}]
    Factoring fractional ideals of $K$ into a product of prime ideals of $\mathcal{O}$ reduces to factoring integers in polynomial time in $\log\Delta$ and $n$.
\end{prop}

A fractional ideal $I$ is given as $d\in\mathcal{O}$ and $A$, the integer which makes $dI$ an integral ideal and a matrix of a basis of $\mathcal{O}$, respectively.

\begin{enumerate}
    \item Compute the norm $N=\mathcal{N}(dI)$.
    \item Factor $N=\prod p^{e_p}$ with $e_p>0$.
    \item For each $p$ dividing $N$, compute the prime ideals $\mathfrak{p}_1,\dots ,\mathfrak{p}_k$ above $p$.
    \item For each $p$ dividing $N$, and each $\mathfrak{p}\supset p\mathcal{O}$ from Step (3), compute $v_{\mathfrak{p}}(dI)$, giving the exponent of $\mathfrak{p}$ in the factorization of $dI$.
    \item For each $\mathfrak{p}$ found with nonzero valuation, output $\mathfrak{p}$, $v_{\mathfrak{p}}(dI)$.
    We have $dI=\prod\mathfrak{p}^{v_{\mathfrak{p}}(dI)}$.
    \item Repeat steps (1)-(5) for the integral ideal $d\mathcal{O}$, then subtract the exponents of $d\mathcal{O}$ from the exponents computed above for the ideal $dI$ for each prime, giving the primes $\mathfrak{p}$ and the exponents $v_{\mathfrak{p}}$ such that $I=\prod\mathfrak{p}^{v_{\mathfrak{p}}}$.
\end{enumerate}

By the multiplicative rule of the ideal norms, for $\mathfrak{a}=\prod\mathfrak{p}^{v_{\mathfrak{p}(\mathfrak{a})}}$, we can factor its norm into $\mathcal{N}(\mathfrak{a})=\prod\mathcal{N}(\mathfrak{p})^{v_{\mathfrak{p}(\mathfrak{a})}}$.
Therefore the norms of $d\mathfrak{a}$ and $d$ will be $\mathcal{N}(d\mathfrak{a})=\prod\mathcal{N}(\mathfrak{p})^{\max\{ 0, v_{\mathfrak{p}(\mathfrak{a})}\}}$ and $\mathcal{N}(d)=\mathcal{N}(d\mathcal{O})=\prod\mathcal{N}(\mathfrak{p})^{-\min\{ 0, v_{\mathfrak{p}(\mathfrak{a})}\}}$, respectively.
Hence, if we only do Step 2 in quantum, which factors positive integers $\mathcal{N}(d\mathfrak{a})$ and $\mathcal{N}(d\mathcal{O})$, then by~\cite{Shor94}, the number of qubits used for ideal factorization will be $O\left(\log\left(\mathcal{N}(d)\cdot (\mathcal{N}(\mathfrak{a})+1)\right)\right)=O(\log\mathcal{N}(\mathfrak{a}))$.
By Theorem~\ref{thm1}, to compute the $S_{\text{PIP}}$-unit group costs $O(n^5+n^4\log\Delta+n^4\log\mathcal{N}(\mathfrak{a}))+O(n\log\tau^{-1})$ qubits, and hence it turns out to be the quantum space complexity for PIP as claimed in Corollary~\ref{corPIP}.

\medskip
\printbibliography

\end{document}